\documentclass[12pt]{article}

\def\UseBibLatex{1}

\makeatletter
\def\input@path{{styles/}}
\makeatother

\usepackage{amsmath}

\usepackage[amsmath,thmmarks]{ntheorem}%

\numberwithin{figure}{section}%
\numberwithin{table}{section}%
\numberwithin{equation}{section}%

\usepackage{hyperref}%
\usepackage[ocgcolorlinks]{ocgx2}

\hypersetup{%
      unicode,
      breaklinks,%
      colorlinks=true,%
      urlcolor=[rgb]{0.25,0.0,0.0},%
      linkcolor=[rgb]{0.5,0.0,0.0},%
      citecolor=[rgb]{0,0.2,0.445},%
      filecolor=[rgb]{0,0,0.4},
      anchorcolor=[rgb]={0.0,0.1,0.2}%
}

\usepackage{caption}
\usepackage{graphicx}
\usepackage{eepic}
\usepackage{xspace}%

\providecommand{\BibLatexMode}[1]{}
\providecommand{\BibTexMode}[1]{}

\renewcommand{\BibLatexMode}[1]{#1}
\renewcommand{\BibTexMode}[1]{}

\ifx\UseBibLatex\undefined%
  \renewcommand{\BibLatexMode}[1]{}
  \renewcommand{\BibTexMode}[1]{#1}
\fi

\BibLatexMode{%
   \usepackage[bibencoding=utf8,style=alphabetic,backend=biber]{biblatex}%
   \usepackage{sariel_biblatex}%
}

\providecommand{\remove}[1]{}

\renewcommand{\Re}{{\rm I\!\hspace{-0.025em} R}}

\def\A{{\cal A}}
\def\L{{\cal L}}
\def\R{{\cal R}}
\def\S{{\cal S}}
\def\T{{\cal T}}
\def\Z{{\cal Z}}
\newcommand{\g}[1]{\gamma_{#1}}

\def\blankline{\hbox{}}

\newcommand{\bi}{\begin{itemize}}
\newcommand{\ei}{\end{itemize}}
\newcommand{\be}{\begin{enumerate}}
\newcommand{\ee}{\end{enumerate}}

\theoremseparator{.}%

\theoremstyle{plain}%
\newtheorem{theorem}{Theorem}[section]

\newtheorem{lemma}[theorem]{Lemma}

\newtheorem{corollary}[theorem]{Corollary}

\theoremstyle{plain}%
\theoremheaderfont{\sf} \theorembodyfont{\upshape}%
\newtheorem*{remark:unnumbered}[theorem]{Remark}%
\newtheorem{remark}[theorem]{Remark}%

\newtheorem{defn}[theorem]{Definition}

\theoremheaderfont{\em}%
\theorembodyfont{\upshape}%
\theoremstyle{nonumberplain}%
\theoremseparator{}%
\theoremsymbol{\myqedsymbol}%
\newtheorem{proof}{Proof:}%

\newtheorem{proofof}{Proof of\!}%

\newenvironment{dfn}{{\vspace*{1ex} \noindent \bf Definition }}{\vspace*{1ex}}

\newcommand{\choosex}[2]{{\binom{#1}{#2}}}

\newcommand{\DBLprime}{{\prime\prime}}
\newcommand{\DS}{{{Davenport-Schinzel}} }

\newcommand{\pth}[1]{\left( {#1} \right)}
\newcommand{\ceil}[1]{\left\lceil {#1} \right\rceil}
\newcommand{\floor}[1]{\left\lfloor {#1} \right\rfloor}
\newcommand{\MakeBig}{\rule[-.2cm]{0cm}{0.4cm}}
\newcommand{\MakeBigger}{\rule[-.25cm]{0cm}{0.75cm}}

\newcommand{\sep}[1]{\,\left|\, {#1} \MakeBig\right.}
\newcommand{\brc}[1]{\left\{ {#1} \right\}}

\newcommand{\HLink}[2]{\hyperref[#2]{#1~\ref*{#2}}}
\newcommand{\HLinkSuffix}[3]{\hyperref[#2]{#1\ref*{#2}{#3}}}

\newcommand{\figlab}[1]{\label{fig:#1}}
\newcommand{\figref}[1]{\HLink{Figure}{fig:#1}}

\newcommand{\thmlab}[1]{{\label{theo:#1}}}
\newcommand{\thmref}[1]{\HLink{Theorem}{theo:#1}}

\newcommand{\corlab}[1]{\label{cor:#1}}

\newcommand{\lemlab}[1]{\label{lemma:#1}}
\newcommand{\lemref}[1]{\HLink{Lemma}{lemma:#1}}%

\providecommand{\eqlab}[1]{}%
\renewcommand{\eqlab}[1]{\label{equation:#1}}
\newcommand{\Eqref}[1]{\HLinkSuffix{Eq.~(}{equation:#1}{)}}

   \theoremheaderfont{\em}%
\providecommand{\etal}{et~al.\xspace}
\renewcommand{\etal}{et~al.\xspace}

\newcommand{\remlab}[1]{\label{rem:#1}}
\newcommand{\remref}[1]{\HLink{Remark}{rem:#1}}%
\newcommand{\seclab}[1]{\label{sec:#1}}
\newcommand{\secref}[1]{\HLink{Section}{sec:#1}}
\newcommand{\tblref}[1]{\HLink{Table}{table:#1}}

\newcommand{\tbllab}[1]{\label{table:#1}}

\usepackage[cm]{fullpage}

\newlength{\CharHeight}
\AtBeginDocument{\setlength{\CharHeight}{\fontcharht\font`X}}   %

   \providecommand{\myqedsymbol}{}
   \renewcommand{\myqedsymbol}{%
      \reflectbox{\includegraphics[height=1.4\CharHeight]%
         {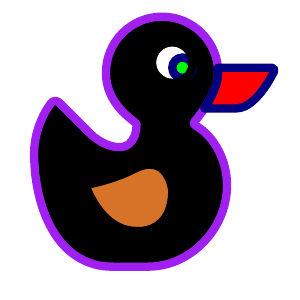}}%
   }%

\BibLatexMode{%
   \bibliography{msc_article}
}

\begin{document}

\title{The Complexity of One or Many Faces in the Overlay of Many Arrangements\footnote{This article is based on the author's master thesis in Tel-Aviv university under the guidance of Micha Sharir. It appeared in journal form in \cite{h-mcl-99}.}%
}%
\author{Sariel Har-Peled%
   \thanks{School of Mathematical Sciences, Tel-Aviv University.}%
}

\date{May 8, 1995\thanks{\LaTeX{}ed on \today.}}
\maketitle

\begin{abstract}
    We present an extension of the Combination Lemma of \cite{GSS89} that expresses the complexity of one or several faces in the overlay of many arrangements, as a function of the number of arrangements, the number of faces, and the complexities of these faces in the separate arrangements.  Several applications of the new Combination Lemma are presented: We first show that the complexity of a single face in an arrangement of $k$ simple polygons with a total of $n$ sides is $\Theta(n \alpha(k) )$, where $\alpha(\cdot)$ is the inverse of Ackermann's function.  We also give a new and simpler proof of the bound $O \pth{ \sqrt{m} \lambda_{s+2}( n ) }$ on the total number of edges of $m$ faces in an arrangement of $n$ Jordan arcs, each pair of which intersect in at most $s$ points, where $\lambda_{s}(n)$ is the maximum length of a \DS sequence of order $s$ with $n$ symbols. We extend this result, showing that the total number of edges of $m$ faces in a sparse arrangement of $n$ Jordan arcs is $O \left( (n + \sqrt{m}\sqrt{w}) \frac{\lambda_{s+2}(n)}{n} \right)$, where $w$ is the total complexity of the arrangement. Several other applications and variants of the Combination Lemma are also presented.
\end{abstract}

\section{Introduction}
\seclab{introduction:thesis}

Let $\Gamma$ be a given set of $n$ Jordan arcs%
\footnote{A {\em Jordan arc} is the image of a continuous $1\!\!-\!\!1$ mapping of the unit interval $[0,1]$ into the plane.}%
in the plane, so that any pair of arcs from $\Gamma$ intersect at most $s$ times (for some fixed constant $s$).  Let $A(\Gamma)$ denote the {\em arrangement} of $\Gamma$, namely the partition of the plane induced by the arcs of $\Gamma$ into $O(n^{2})$ faces, edges and vertices.  The vertices are the endpoints and the points of intersection of the arcs in $\Gamma$, the edges are the maximal connected portions of the arcs not containing any vertex, and the faces are the connected components of the complement of the union of the arcs in $\Gamma$.  Such an arrangement is depicted in \figref{arrangement}. See \cite{sa-dsstg-95} for more details concerning planar arrangements of arcs.

\begin{figure}
    \centering%
    \includegraphics{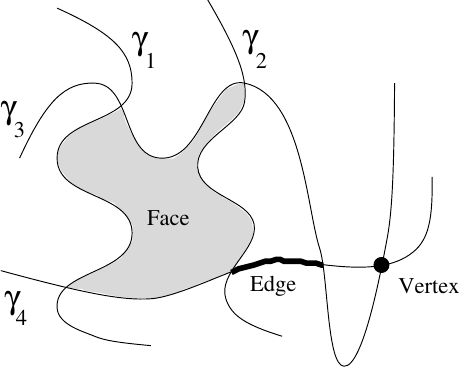}
    \caption{The arrangement of $\brc{ \gamma_1,
      \gamma_2, \gamma_3, \gamma_4 }$.}
    \figlab{arrangement}
\end{figure}

For a face $F$ of the arrangement $A(\Gamma)$, the {\em complexity} of $F$ is the number of vertices and edges (subarcs) along the boundary of $F$.

Let $\Gamma_{1}, \ldots, \Gamma_{t}$ be $t$ sets of Jordan arcs, so that any pair of arcs from $\Gamma = \bigcup_{i=1}^{t} \Gamma_{i}$ intersect at most $s$ times. The arrangement $A( \Gamma )$ is called the {\em overlay} of the arrangements $A(\Gamma_{1}), \ldots, A(\Gamma_{t})$.  Given a set of points $P = \{ p_{1}, \ldots, p_{k} \}$, none lying on any arc of $\Gamma$, we call a face, of any arrangement, {\em marked} if it contains at least one point from $P$.

We are interested in the complexity of the marked faces in $A( \Gamma )$, expressed as a function of the total complexity of the marked faces in the arrangements $A(\Gamma_{1}), \ldots, A(\Gamma_{t})$.  The case $t=2$ has been studied by Guibas \etal.  \cite{GSS89}, where the following result was proved:

\begin{lemma}
    \lemlab{std_combination}%
    Given two sets $\Gamma_{1}, \Gamma_{2}$ of Jordan arcs as above, and a set $P$ of $k$ marking points, the total complexity of all the marked faces in $A(\Gamma_{1} \cup \Gamma_{2})$ is $O( r + b + k)$, where $r, b$ are the total complexities of all marked faces in $A(\Gamma_{1}), A(\Gamma_{2})$, respectively.
\end{lemma}

\begin{figure}[ht]
    \centering
    \includegraphics{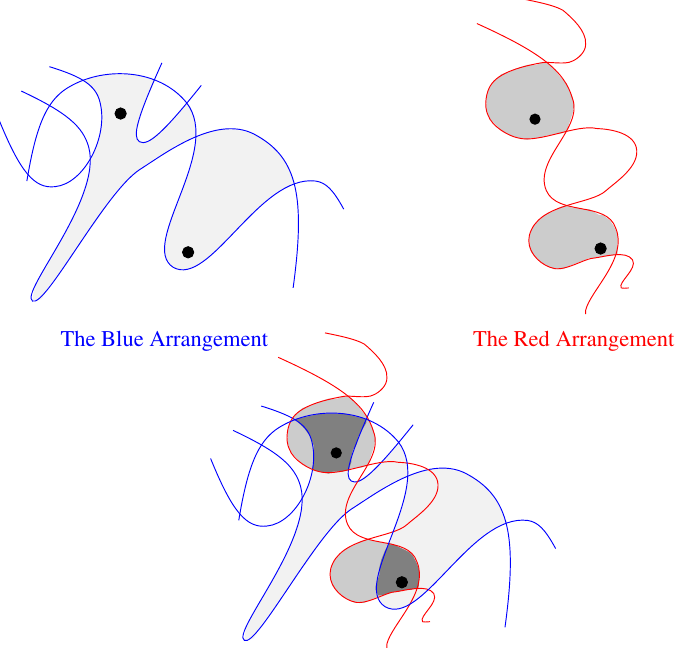}
    \caption{Overlay of two arrangements with two
       marked faces}
    \figlab{overlay}
\end{figure}

See \figref{overlay} for an example of an overlay of two arrangements.

The constant in the bound stated in the lemma is $s+3$.  This is already the case when $P$ contains just a single point (i.e., only a single face is marked), where we have the following more detailed bound:

\begin{lemma}[\cite{GSS89}]
    The complexity of a single marked face $E$ in $A(\Gamma_{1} \cup \Gamma_{2})$ is at most $(s+3)(b + r+ 2u +2v -4t)$ where $b, r$ are as above, $u, v$ are the numbers of connected components of the boundary of the marked faces in $A(\Gamma_{1}), A(\Gamma_{2})$, respectively, and $t$ is the number of connected components of $\partial{E}$.
\end{lemma}

In this paper we consider the more general case of overlaying $t>2$ arrangements, $A( \Gamma_{1} ), \ldots, A(\Gamma_{t})$, as above, where we want to bound the overall complexity of the marked faces in $A( \Gamma_{1} \cup \cdots \cup \Gamma_{t} )$ as a function of the total complexity of the marked faces in $A(\Gamma_{1}), \ldots, A(\Gamma_{t})$, of $t$, the number of arrangements, and of $k$, the number of marking points.  One way of doing this is to overlay the given arrangements two at a time, in a balanced binary-tree fashion, and apply the bounds stated above.  However, this results in a bound that is too large, because the bound is multiplied by a factor of $s+3$ at each overlay step, leading to an overall factor of $(s+3)^{\log_2{t}} = t^{\log_2{(s+3)}}$, which is much too large. We derive an improved, and nearly worst-case tight, bound as follows.

First we deal with the simpler case of lower envelopes. That is, we are given $t$ lower envelopes of $t$ respective collections of continuous functions defined over the reals, any pair of which has at most $s$ intersection points, and we want to bound the complexity of the lower envelope of these $t$ envelopes.  We show, in \secref{lower_envelope_section}, that this complexity is $O \left( \frac{\lambda_{s}(t)}{t} C \right)$, where $C$ is the overall complexity of the individual lower envelopes. Here $\lambda_{s}(t)$ denotes the maximum length of $(t,s)$-\DS sequences ($DS(t, s)$ sequences in short), i.e., sequences composed of $t$ symbols, which do not contain equal adjacent elements and also do not contain any alternating subsequence of two distinct symbols of length $s+2$. It is shown in \cite{HS86}, \cite{ASS89} that $\lambda_{s}(t)$ is almost linear in $t$ for any fixed $s$.  See \cite{a-sdcgp-85, ASS89,HS86, sa-dsstg-95} for a more detailed review of \DS sequences. We also show that this bound is tight in the worst case.

In \secref{restricted:sequences} we prove a bound on the complexity of a \DS sequence as a function of its structure when restricted to subsets of its symbols. This result extends the result on the complexity of the lower envelope of multiple lower envelopes.  We call a symbol $a$ {\em active} in a \DS sequence $S = (s_1, s_2, \ldots, s_q)$ in position $i$, for $1 \leq i \leq q$, if there are $1 \leq u \leq i < v \leq q$, such that $s_u = s_v = a$. We give a new and simpler proof of a result from \cite{hkk-vdrms-92}, stating that the length of a $DS(n, s)$ sequence, such that in each position there are at most $k$ active symbols, is $O \left( \frac{\lambda_{s+2}( k ) }{k} n \right)$.

In \secref{single:face:overlay} we use the result on restricted \DS sequences to show that the complexity of a single face (marked by a point) in an overlay of $t$ arrangements, as above, is $O( \frac{\lambda_{s+2}(t)}{t} C )$, where $C$ is the total complexity of the $t$ marked faces in the $t$ arrangements.  Again, this bound is shown to be tight in the worst case. As a simple application of this bound, we show that the maximum complexity of a single face in an arrangement of $k$ simple polygons with a total of $n$ sides is $\Theta( n\alpha(k))$.

In \secref{line:seg:arrangement} we present new bounds on the complexity of a single face in an arrangement of line segments, for certain special classes of segments. For example, we derive a linear bound on the complexity of a single face, in the case where the endpoints of the line segments all lie on a fixed circle and the line segments are ``large'' relative to the radius of the circle.

In \secref{many:faces:overlay} we prove our main result, showing that the complexity of $k$ marked faces in the overlay of $t$ arrangements is $O \left( \frac{\lambda_{s+2}(t)}{t} C + k \lambda_{s+2}(t) \right)$, where $C$ is the total complexity of all the marked faces in the $t$ arrangements. This bound is shown to be close to tight in the worst case.

We also present some applications of our results. We derive a new proof, of the bound $O \left( \sqrt{m} \lambda_{s+2}(n) \right)$ on the complexity of $m$ faces in the arrangement of $n$ Jordan arcs. Our proof is simpler than the original proof given in \cite{EGPPSS92}.

In \secref{sparse:arrangement} we apply our main theorem to derive the bound \\ $O \left( \frac{\lambda_{s+2}(k)}{k} n + m \lambda_{s+2}(k) \right)$, on the complexity of $m$ faces in an arrangement of $n$ Jordan arcs, where $k$ is the {\em chromatic number} of the arrangement. The chromatic number is the minimum number of colors needed to color the arcs of the arrangement, such that no pair of intersecting arcs are colored by the same color.  As a simple corollary, we show that the complexity of $m$ faces in an arrangement of $n$ Jordan arcs, whose overall complexity $w$, is $O \left( \frac{\lambda_{s+2}(\sqrt{w})}{\sqrt{w}} n + m \lambda_{s+2}(\sqrt{w}) \right)$. Using more sophisticated analysis, we improve this bound to $O \left( (n + \sqrt{m}\sqrt{w}) \frac{\lambda_{s+2}(n)}{n} \right)\,$. If $w =\Omega \pth{ n^{2} }$, this bound coincides with the bound $O \pth{ \sqrt{m} \lambda_{s+2}(n) }$ mentioned above, but it is considerably smaller when $w \ll n^2$.

\section{The Complexity of a Single Face in an Overlay of Arrangements}
\seclab{single_face}

In this section, we derive a tight bound on the maximum complexity of a single face in an overlay of many arrangements of arcs. This bound is used in \secref{many:faces:overlay} to prove our main theorem. We also derive several related results, which we believe to be interesting in their own right.

\subsection{Introduction}
\seclab{face:introduction}

Let $\Gamma$ be a collection of Jordan arcs, any pair of which intersect at most $s$ times. Let $\Gamma_{1}, \ldots, \Gamma_{t}$ be a partition of $\Gamma$ into $t$ disjoint subsets.  Let $A^{1} = A(\Gamma_{1}), \ldots, A^{t} = A(\Gamma_{t})$ denote the arrangements of the arcs in $\Gamma_{1}, \ldots, \Gamma_{t}$, respectively. Given a point $p$, we denote by $f_{i} = f_{i}(p)$ the face in $A^{i}$ that contains $p$, for $1 \leq i \leq t$, and let $C_{i}$ denote the combinatorial complexity of $f_{i}$ (i.e., the number of edges and vertices of $A^{i}$ that appear along $\partial{f_{i}}$).  Let $C = \sum_{i=1}^{t} C_{i}$ denote the total complexity of all these faces. Let $A = A(\Gamma)$ be the overlay arrangement formed by all the arcs in $\Gamma$.

\begin{theorem}[Single Face Combination Theorem]
    \thmlab{complex:Face:In:Overlay}%
    The complexity of the face in $A$ that contains $p$ is $O \pth{ \frac{\lambda_{s+2}(t)}{t} C }$, and this bound is tight in the worst case.
\end{theorem}

We first prove the lower bound of the theorem.  Let $\Gamma_{0} = \{ \gamma_{1}, \ldots, \gamma_{t} \}$ be a collection of $t$ Jordan arcs, all contained in some fixed vertical strip, such that any pair of arcs of $\Gamma_{0}$ intersect in at most $s$ points, and such that the complexity of the lower envelope of the arcs in $\Gamma_{0}$ is $\lambda_{s+2}(t)$ (see \cite{HS86}).

Let $a$ denote the width of the strip containing $\Gamma_{0}$, and let $\Gamma_{i}$, for $i=1, \ldots, t$, denote the set
\[
    \brc{ \gamma_{i} + j a \sep{ j = 0, \ldots, m - 1 } },
\]
for some integer $m$, where $\gamma + x$ denotes the copy of $\gamma$ obtained by shifting $\gamma$ by $x$ units in the $x$-direction.  Let $\Gamma = \cup_{i=1}^{t} \Gamma_{i}$, and let $p$ be any point lying below the lower envelope of $\Gamma$.

In this case we have $t$ arrangements, each consisting of pairwise-disjoint arcs. Hence the total complexity $C$ of the faces containing $p$ in the separate arrangements is proportional to the total number of arcs, that is, $C = \Theta(m t)$. The complexity of $f$ is
\[
    \Omega \pth{ \lambda_{s+2}(t) m } = \Omega \pth{ \frac{\lambda_{s+2}(t)}{t} C },
\] %
which establishes the lower bound of the lemma.

Before establishing the upper bound, we first deal, in \secref{lower_envelope_section}, with the simpler case of lower envelopes. Motivated by this result, we prove a corresponding bound on the length of certain restricted \DS sequences in \secref{restricted:sequences}.  We then use this latter bound in \secref{single:face:overlay} to prove  \thmref{complex:Face:In:Overlay}.

\subsection{The Complexity of the Overlay of Lower Envelopes}
\seclab{lower_envelope_section}

As a warm-up exercise, we consider the simpler case of lower envelopes.  Let $F_{1}, \ldots, F_{t}$ be $t$ collections of continuous functions over the reals, and let $F = \bigcup_{i=1}^{t} F_{i}$ be their union. Assume that each pair of functions in $F$ intersect at most $s$ times.  Let $S_{i}$ be the lower envelope of the functions in $F_{i}$, and let $C_{i}$ denote the complexity of $S_{i}$, for $1 \leq i \leq t$.  Let $C$ denote the total complexity of the lower envelopes of $F_{1}, \ldots, F_{t}$, i.e., $C = \sum_{i=1}^{t}C_{i}$.

\begin{lemma}
    \lemlab{lower_env}%
    The complexity of the lower envelope of the functions in $F$ is $O \pth{ \frac{\lambda_{s}(t)}{t} C }$, and this bound is tight in the worst case.
\end{lemma}

\begin{proof}
    The lower bound can be established in a manner similar to that given above, so we only prove the upper bound.  Let $S$ denote the lower envelope of $F$.  Without loss of generality we can assume that all the lower envelopes $S_{1}, \ldots, S_{t}$ appear along $S$, i.e., for every $1 \leq i \leq t$ there exists $x_{i}$ where $S_{i}$ is the lowest of the $t$ envelopes at $x_{i}$.  If this does not hold, we can discard all the unused lower envelopes, and argue that %
    $|S|=O \pth{ \frac{\lambda_{s}(l)}{l} C }$, where $l < t$ is the number of lower envelopes that actually appear along $S$.  Since $\frac{\lambda_{s}(l)}{l} \leq \frac{\lambda_{s}(t)}{t}$, it follows that %
    $|S| = O \pth{ \frac{\lambda_{s}(t)}{t} C }$.

    \begin{figure}[ht]
        \centering
        \includegraphics{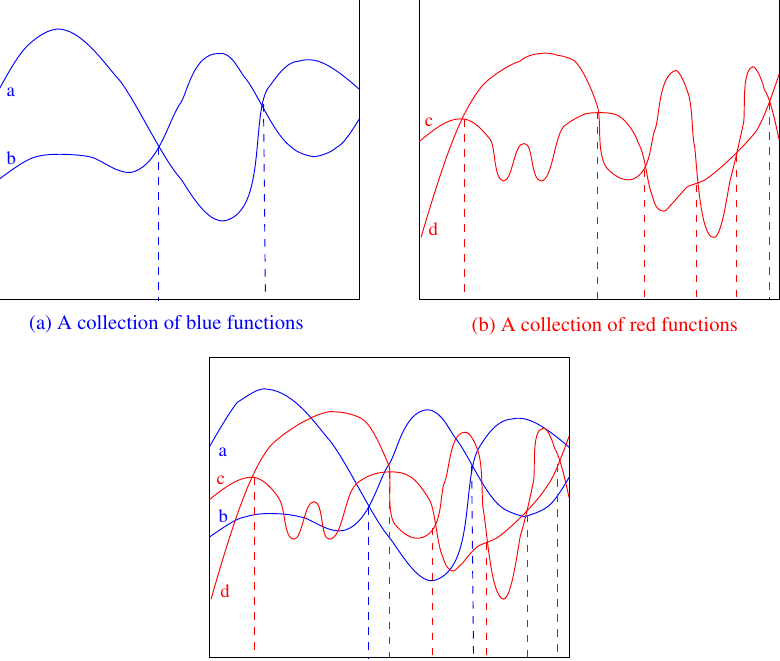}
        \caption{Overlay of the lower envelopes of the red and blue functions}
        \figlab{lower:envelope}
    \end{figure}

    Let $T_{i}$ be the set of transition points (breakpoints) in the lower envelope $S_{i}$, for $1 \leq i \leq t$, where a transition point of $S_{i}$ is the abscissa of an intersection point of two functions of $F_{i}$ that appears along $S_{i}$.

    Let $T = \bigcup_{i=1}^{t} T_{i}$. A maximal interval $I$ whose interior does not contain any point of $T$ is called {\em atomic}.  Thus no lower envelope $S_{1}, \ldots, S_{t}$ has a transition point in the interior of any atomic interval (See \figref{lower:envelope}). It follows that any transition point of $S$ appearing in the interior of an atomic interval is an intersection point between a pair of distinct lower envelopes $S_{i}, S_{j}$.

    Let $I = [a,b]$ be a fixed atomic interval, and let $Y_{a}, Y_{b}$ be the two vertical lines passing through $a, b$, respectively. We wish to bound the complexity of the lower envelope $S$ restricted to the strip between $Y_{a}$ and $Y_{b}$.  By construction, there are at most $t$ functions that may participate in the lower envelope in this strip, one from each collection $F_{i}$, for $1 \leq i \leq t$.

    Consider any two adjacent strips defined by atomic intervals in $T$, and observe that the sets of functions that may participate in the lower envelope over each of these strips differ by exactly one function in the first set being replaced by another function in the second set.  (We assume here general position of the given envelopes; otherwise some straightforward modifications of the following arguments are required.)  Thus, for a group of $k$ adjacent strips, at most $t + k - 1$ functions participate in the lower envelope $S$ over these strips. Thus the complexity of the lower envelope over these strips is at most $\lambda_{s}( t + k -1 )$.

    Now partition the set of strips defined by $T$ into groups, each consisting of $k$ adjacent strips (extending the last group, if necessary, so that it contains at most $2k-1$ strips). We have $\max \pth{ \floor{ \frac{|T| + 1}{k} }, 1 }$ such groups. Thus the complexity of the lower envelope is bounded by
    \[ |S| \leq \max \pth{ \floor{ \frac{|T| + 1}{k} }, 1 } \cdot \lambda_{s}( t + 2k -2 ).
    \]

    Letting $k=t$ and observing that (a) $\lambda_{s}( 3t -2 ) = O(\lambda_{s}(t))$, (b) $|T| + 1 \leq C$, and (c) $C \geq t$, we obtain %
    \[
        |S| = O \pth{ \frac{C}{t} \lambda_{s}(t) } = O \pth{ \frac{\lambda_{s}(t)}{t} C },
    \]
    thus completing the proof of the lemma.
\end{proof}

\begin{remark}
    It is easily verified that the above analysis also extends to the case of partially-defined functions. In this case, the complexity of the lower envelope is at most $O \pth{ \frac{\lambda_{s+2}(t)}{t} C }$.
\end{remark}

\subsection{Restricted \DS Sequences}
\seclab{restricted:sequences}

In this section, we prove a bound on the length of a \DS sequence, expressed in terms of the structure of the sequence restricted to subsets of its symbols. If $S$ is a sequence and $A$ is a subset of its symbols, we denote by $S_{A}$ the subsequence of $S$ obtained by removing all elements not in $A$, and by replacing maximal runs of adjacent equal elements by a single element.

\begin{theorem}
    \thmlab{restrict:DS} Let $S$ be a $DS( n, s )$ sequence whose symbols belong to a set $T$ of size $n$. Suppose that $T$ is the disjoint union of $k$ subsets $T_1, \ldots, T_k$, such that $|S_{T_1}| = C_1, \ldots, |S_{T_k}| = C_k$, Then we have
    \[
        |S| = O \pth{ C \frac{\lambda_{s}(k)}{k} },
    \]
    where $C = \sum_{i=1}^{k} C_i$.
\end{theorem}

\begin{proof}
    Without loss of generality, we may assume that all the restricted sequences $S_{T_i}$ are not empty, for $1 \leq i \leq k$.

    Let $S = (s_1, \ldots, s_r)$ be the given sequence. And let $S_1 = S_{T_1}, S_2 = S_{T_2}, \ldots, S_k = S_{T_k}$ denote the restricted sequences, generated by restricting $S$ to $T_1, \ldots, T_k$, respectively. Let $s( i, j )$, for $1 \leq i \leq k, 1 \leq j \leq |S_i|$, denote the symbol appearing at position $j$ in the sequence $S_i$. Let $p = p( i,j)$ denote the last position in $S$ such that $s_p = s(i,j)$ and when restricting $S$ to $T_i$, this position is being mapped to position $j$ of $S_i$.

    We break $S$ into contiguous subsequences by adding delimiters to $S$ as follows. We scan each restricted sequence $S_i$, for $i = 1, \ldots, k$, and insert a delimiter for each symbol $s( i, j )$, for $1 \leq i \leq k, 1 \leq j < |S_i|$, between the positions $p(i,j)$ and $p(i,j)+1$ in $S$.

    Let $\sigma^{1}, \ldots, \sigma^{m}$ denote the maximal contiguous subsequences of $S$ not containing any delimiter, in their order along $S$. Since the total number of delimiters is $C-k$ we have $m= C -k + 1 \leq C$.

    Partition $S$ into disjoint blocks, where each block, except the last one is the concatenation of $k$ consecutive subsequences $\sigma^{j}$; the last block may consist of up to $2k-1$ subsequences. We claim that each block is a $DS(3k, s)$ sequence.  This follows from the fact that only $k$ distinct symbols, one from each subset $T_i$, may appear in any single subsequence $\sigma^{j}$, and that the set of symbols that appear in $\sigma^{j}$ differs from the set of symbols that appear in $\sigma^{j+1}$ by exactly one symbol.  Hence the total length of $S$ is
    \[
        |S| = O \pth{ \floor{ \frac{m}{k} } \lambda_{s}(3k) } = O \pth{ \frac{C}{k} \lambda_s(k) } = O \pth{ C \frac{\lambda_s(k)}{k} }.
    \]
\end{proof}

We present two applications to \thmref{restrict:DS}.  First we give a second proof to \lemref{lower_env}:

\begin{proofof}
    \lemref{lower_env}: Let $S$ be the lower envelope sequence of $F$, and let $S_1, \ldots, S_k$ denote the lower envelope sequences of $F_1, \ldots, F_k$, respectively. It can be easily checked that $S_{F_i}$ is a subsequence of $S_i$, for $i=1, \ldots, k$. Thus the total length of the restricted sequences of $S$ over $F_1, \ldots, F_k$ is bounded by $C$.

    We can therefore apply \thmref{restrict:DS} to $S$, $F$, $F_1, \ldots, F_k$, concluding that \\ $|S| = O \pth{ C \frac{\lambda_{s}(k)}{k} }$ (or, for partially-defined functions, $|S| = O \pth{ C \frac{\lambda_{s+2}(k)}{k} })$, as claimed.
\end{proofof}

\blankline

\begin{defn}
    Let $S = (s_1, \ldots, s_r)$ be a $DS(n, s)$ sequence.  A symbol $a$ is {\em active} in position $j$ in $S$ if there are two indices $1 \leq j^{\prime} \leq j < j^{\DBLprime} \leq r$ such that $s_{j^{\prime}} = s_{j^{\DBLprime}} = a$. Let $v(j)$ denote the number of active symbols in position $j$ of $S$. We call a sequence $S$ a {\em $DSA( n, s, k)$ sequence} if $S$ is a $DS( n, s)$ sequence and, for each $1 \leq j \leq |S|$, we have $v(j) \leq k$. Let $\lambda_{s,k}(n)$ denote the maximum length of a $DSA(n, s, k)$ sequence.
   \end{defn}

We give a new proof to the following lemma, originally established in \cite{hkk-vdrms-92}:

\begin{lemma}
    \lemlab{klara}%
    $\lambda_{s,k}(n) = O( n \frac{\lambda_s(k)}{k} )$.
\end{lemma}

\begin{proof}
    Let $S$ be a given $DSA(n,s, k)$ sequence, and let $T$ denote the set of its symbols. Assume without loss of generality that each symbol of $T$ appears in $S$ at least twice. We partition $T$ into $k$ subsets. We initialize all subsets $T_i$ to the empty set, and flag each of them as `free'.

    We scan $S$ from left to right. Each time we encounter a symbol $a$ for the first time, we allocate a free subset $T_i$, store $a$ in $T_i$ and flag $T_i$ as `active'.  Each time we encounter the last appearance of a symbol $a$ in $S$, we free the subset that has stored $a$. Since there are at most $k$ active symbols, at any position of $S$, it is clear that this process will not get stuck.

    Clearly, each symbol of any of the restricted sequences $S_{T_1}, \ldots, S_{T_k}$, appears there exactly once. Thus, the total length of these sequences is $n$. Applying \thmref{restrict:DS}, it follows that $|S| = O \pth{ n \frac{\lambda_s(k)}{k} }$.
\end{proof}

\subsection{The Complexity of a Single Face in an Overlay of %
   Arrangements}
\seclab{single:face:overlay}

In this section, we finally prove the upper bound of \thmref{complex:Face:In:Overlay}.

Let $\Gamma = \{ \gamma_1, \ldots, \gamma_n \}$ be a collection of $n$ Jordan arcs in the plane, such that any pair of arcs from $\Gamma$ has at most $s$ intersection points.

Let $f$ be any face of $A(\Gamma)$ and let $\zeta$ be any boundary component of $f$. Without loss of generality, we may assume that $\zeta$ is the external boundary component of $f$.  Traverse $\zeta$ in counterclockwise direction (so that $f$ lies to our left) and let $S = (s_1, s_2, \ldots, s_t)$ be the circular sequence of oriented curves in $\Gamma$ in the order in which they appear along $\zeta$.  More precisely, for each $\gamma_i \in \Gamma$, let $u_i$ and $v_i$ be the endpoints of $\gamma_i$.  Assume that $\gamma_i$ is positively oriented from $u_i$ to $v_i$ (and negatively oriented in the other direction).  Denote the positively oriented curve by $\gamma_{i}^{+}$ and the negatively oriented curve by $\gamma_{i}^{-}$.  If during our traversal of $\zeta$ we encounter a curve $\gamma_{i}$ and follow it in the direction from $u_i$ to $v_i$ (respectively, from $v_i$ to $u_i$), then we add $\gamma_{i}^+$ (respectively, $\gamma_i^-$) to $S$. As an example, if the endpoint $u_i$ of $\gamma_i$ is on $\zeta$ and is not incident to any other arc, then traversing $\zeta$ past $u_i$ will add the pair of elements $\gamma_i^-, \gamma_i^+$ to $S$, and symmetrically for $v_i$.  Note that in this example both sides of an arc $\gamma_i$ might belong to our connected component.

We denote the oriented arcs of $\Gamma$ as $\xi_1, \ldots, \xi_{2n}$. For each $\xi_i$ we denote by $|\xi_i|$ the non-oriented arc $\gamma_j$ coinciding with $\xi_{i}$.  We will use the following results of \cite{GSS89}:

\begin{lemma}[The Consistency Lemma \cite{GSS89}]

    The portions of each arc $\xi_{i}$ appear in $S$ in a circular order that is consistent with their order along the oriented $\xi_{i}$; that is, there exists a starting point in $S$ (which depends on $\xi_{i}$) such that if we read $S$ in circular order starting from that point, we encounter these portions in their order along $\xi_{i}$.  \lemlab{consist}
\end{lemma}

The following theorem, taken from \cite{GSS89}, gives a tight bound on the complexity of a single face in an arrangement of $n$ arcs:

\begin{theorem}[The Complexity of a Single Face]
    Given a collection \\ $\Gamma= \{ \gamma_{1}, \ldots, \gamma_{n} \}$ of $n$ Jordan arcs, so that every pair of arcs from $\Gamma$ have at most $s$ intersection points, the complexity of a single face in the arrangement $A(\Gamma)$ is $O(\lambda_{s+2}(n))$.  \thmlab{single_Face_Complex}
\end{theorem}

First, here is an outline of the proof of the upper bound of \thmref{complex:Face:In:Overlay}, which somewhat resembles the proof of Theorem 3.1 of \cite{GSS89}:

\begin{itemize}
    \item We analyze each component $\zeta$ of the boundary of $f$ separately.  As in the proof of \lemref{lower_env}, we can assume that arcs from all the $t$ arrangements appear along $\zeta$.  Let $S=(s_{1}, \ldots, s_{q})$ be the circular sequence of arcs appearing along $\zeta$, as defined above.  We apply to $S$ the transformations described in the proof of Theorem 3.1 of \cite{GSS89} (see below for details), after which $S$ becomes a $DS(l, s+2)$ sequence, where $l$ is at most four times the number of original arcs appearing along $\zeta$.

    \item In step (1) below, an extension of the Consistency Lemma shows, for each $1 \leq i \leq t$, that the restricted sequences $R^{(i)} = S_{\Gamma_i}$ of arcs of $\Gamma_{i}$ along $\zeta$ is consistent with the sequences of arcs of $\Gamma_{i}$, along each component of the boundary of the face $f_i$ that contains $f$ in the arrangement $A^{i}$.

    \item Step (2) below shows that the length of $R^{(i)}$ is at most $C_{\zeta}^{i} + 2u_{\zeta}^{i} - 2$, where $C_{\zeta}^{i}$ is the total complexity of the connected components of $\partial{f_{i}}$ appearing along $\zeta$, and $u_{\zeta}^{i}$ is the number of such connected components.  The total length of all the $R^{(i)}$'s is at most $C_{\zeta} + 2u_{\zeta} - 2t = O(C_{\zeta})$, where $C_{\zeta} = \sum_{i=1}^{t} C_{\zeta}^{i}$ and $u_{\zeta} = \sum_{i=1}^{t} u_{\zeta}^{i}$.

    \item Now \thmref{restrict:DS} implies that $|S| = O \pth{\frac{\lambda_{s+2}(t)}{t} C_{\zeta} }$. summing this over all boundary components $\zeta$ completes the proof of the theorem.
\end{itemize}

We now give the proof in detail.  Let us fix a single component $\zeta$ of $\partial{f}$, and assume, without loss of generality, that it is the exterior component. Trace $\zeta$ in counterclockwise direction (with $f$ lying to the left), and let $S=(s_{1}, \ldots, s_{q})$ be the (circular) sequence of the subarcs of $\partial{f_{1}}, \ldots, \partial{f_{t}}$ as they appear along $\zeta$.

\begin{figure}[ht]
    \centering%
    \includegraphics{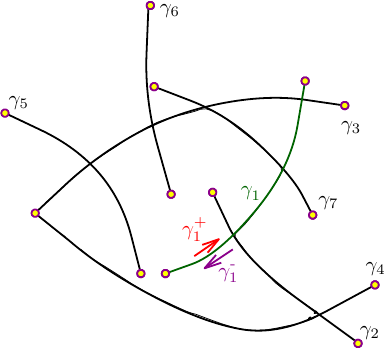}
    \caption{Traversing a face boundary and the resulting sequence %
       $S = \{ \g{1}^{+} ${}$ \g{2}^{-}$\allowbreak{}%
       $\g{2}^{+}$\allowbreak{}$\g{1}^{+}$\allowbreak{}$\g{7}^{-}$%
       \allowbreak{}$\g{3}^{-}$\allowbreak{}$\g{6}^{+}$\allowbreak{}%
       $\g{6}^{-}$\allowbreak{}%
       $\g{3}^{-}$\allowbreak{}%
       $\g{5}^{+}$\allowbreak{}%
       $\g{5}^{-}$\allowbreak{}%
       $\g{3}^{-}$\allowbreak{}%
       $\g{4}^{+} \g{2}^{-} \g{1}^{-} \}$.  }
    \figlab{simp:sequence}
\end{figure}

We apply two transformations to $S$. First we use different symbols for the subarcs traced in the direction of $\gamma_i^{+}$ and for the subarcs traced in the direction of $\gamma^{-}_{i}$ (that is, we use the arcs $\xi_1, \ldots, \xi_2n$ as the symbols of $S$). Second, we apply \lemref{consist} to each $\xi_i$ (See \figref{simp:sequence}). We linearize $S$ so that it starts at $s_1$ and ends at $s_q$. Suppose that the first appearance of $\xi_i$ (as described in \lemref{consist}) is at $s_{\alpha_i}$ and the last appearance is at $s_{\beta_i}$. If $\alpha_i \leq \beta_i$ we do nothing. Otherwise we use two new symbols, one to denote all appearances of $\xi_i$ between $s_{\alpha_i}$ and $s_{\beta_i}$, and one to denote the appearances between $s_1$ and $s_{b_i}$. This done exactly as in the analysis of \cite{GSS89}. With these modifications, $S$ is now a (linear) $DS(l, s+2)$ where $l$ is as above.

By definition, the sum of the lengths of the sequences $S$, over all connected portions $\zeta$ of $\partial{f}$, is the complexity of $f$.

The proof consists of the following steps:

(1) Let $a$ be a subarc of $\partial{f_{i}}$ which appears along $\zeta$. Let $a_{1}, a_{2}$ be two connected portions of $a \cap \zeta$, consecutive along $a$, such that when $a$ is traversed with $f_{i}$ lying to its left, $a_{1}$ precedes $a_{2}$. It follows from the Consistency \lemref{consist} that $a_{1}$ and $a_{2}$ are also adjacent along $\zeta$, in the strong sense that the portion of $\zeta$ between $a_{1}$ and $a_{2}$ does not intersect the connected component of $\partial{f_{i}}$ containing $a$.

(2) We claim that the length of $R^{(i)}$ is at most $C_{\zeta}^{i} + 2u_{\zeta}^{i} - 2$, for $i=1, \ldots, t$.

To prove the claim, assume $a$ is an arc of $\partial{f_{i}}$ appearing more than once in $R^{(i)} = (r^{1}, \ldots, r^{q})$.  By (1), all elements of $R^{(i)}$ lying between two consecutive appearances, $r^{j}, r^{k}$, of $a$ (arranged in this order along $a$) must belong to other components of $\partial{f_{i}}$.  We charge the second appearance of $a$ to the component of $\partial{f_{i}}$ containing $r^{j+1}$. Let $\sigma^{i}$ be the (circular) sequence of the connected components of $\partial{f_{i}}$ in the order that they appear along $\zeta$ (so that no two adjacent elements of $\sigma^{i}$ are equal). As in the proof \cite{EGS90} of the standard Combination \lemref{std_combination}, it is fairly easy to show that $\sigma^{i}$ is a circular $(u_{\zeta}^{i}, 2)$-{Davenport-Schinzel} sequence (i.e., it is composed of $u_{\zeta}^{i}$ symbols, no two adjacent elements of it are equal, and it does not contain a subcycle of the form $(a \cdots b \cdots a \cdots b)$, for any two distinct elements $a,b$).  Hence its length is at most $2u^{i}_{\zeta} - 2$ (see \cite{ES90}). Moreover, it is easily checked that the charging scheme described above never charges an element of $\sigma^{i}$ more than once.  Hence the total number of duplications of elements in $R^{(i)}$ is at most $2u_{\zeta}^{i} - 2$, from which the claim follows.

The total length of the restricted sequences is thus
\begin{eqnarray*} { \sum_{j=1}^{t} |R^{(j)}| } \leq { \sum_{j=1}^{t} {(C_{\zeta}^{j} + 2 u_{\zeta}^{j} -2 )} } \leq C_{\zeta} + 2u_{\zeta} - 2t = O(C_{\zeta}).
\end{eqnarray*}

(3) Now, \thmref{restrict:DS} implies that $|S| = O \pth{ \frac{\lambda_{s+2}(t)}{t} C_{\zeta} }$.

Finally, we sum these inequalities over all components $\zeta$ of $\partial{f}$, and observe that
\[
    \sum_{\zeta \in \partial{f}} C_{\zeta} = C,
\]
because no arc of any sub-arrangement can appear along two distinct components of $\partial{f}$.  This completes the proof of \thmref{complex:Face:In:Overlay}.  {\hfill\myqedsymbol}

\blankline

Here is a simple but useful consequence of \thmref{complex:Face:In:Overlay}:

\begin{lemma}
    The maximum complexity of a single face in the arrangement of $k$ simple polygons having a total of $n$ vertices is $\Theta \pth{ n \alpha(k) }$.

    \lemlab{single:face:polygons}
\end{lemma}

\begin{proof}
    The lower bound has been observed in \cite{AS94}. Concerning the upper bound, let $P_1, \ldots, P_{k}$ be the given polygons. Let $A_i = A( P_i)$, for $i=1, \ldots, k$; that is, $A_i$ is just the polygon $P_i$. Let $A$ denote the overlay arrangement of the arrangements (i.e., polygons) $A_{1}, \ldots, A_{k}$. Note that the sum of the complexities of the individual arrangements is $O(n)$. Also, $s=1$ for collections of segments, and $\frac{\lambda_{s+2}(k)}{k} = \frac{\lambda_{3}{(k)}}{k} = O( \alpha(k) )$. Hence the upper bound of the lemma follows immediately from \thmref{complex:Face:In:Overlay}.
\end{proof}

\begin{remark}
    \lemref{single:face:polygons} has been used in \cite{AS94} for the special case of a single face in the complement of the union of $k$ convex polygons with a total of $n$ vertices.
\end{remark}

\begin{defn}
    Let $A$ and $B$ be two sets in the plane.  The {\em Minkowski sum} (or vector sum) of $A$ and $B$, denoted $A \oplus B$, is the set $\brc{ a + b \sep{ a \in A, b \in B }}$.
\end{defn}

A related consequence of \thmref{complex:Face:In:Overlay} is the following result of \cite{HCAHS95}:

\begin{theorem}
    Let $P$ and $Q$ be polygonal sets with $k$ and $n$ vertices respectively, where $k \le n$. The complexity of a face of the complement of the Minkowski sum $P \oplus Q$ is $O \pth{ nk \alpha(k) }$.

    \thmlab{mink}
\end{theorem}

\begin{proof}
    Each segment that bounds the Minkowski sum $P \oplus Q$ is the Minkowski sums of a vertex of one polygonal set with an edge of the other (See \cite{GRS83}). We treat these asymmetrically, and define a {\em vertex set} to be the sum of a fixed vertex of $P$ with all the edges of $Q$ and an {\em edge set} to be the sum of a fixed edge of $P$ and all the vertices of $Q$.

    Let $p_1, \ldots, p_k$ be the vertices of $P$ and let $e_1, \ldots, e_k$ be the edges of $P$.

    Let $V_{Q,p_i}$ be the vertex set resulting from the contact of $p_i \in P$, for $i = 1, \ldots, k$, with the edges of $Q$. Clearly, $V_{Q, p_i}$ is just a translated copy of $Q$.

    Let $V_{Q, e_j}$ be the edge set resulting from the contact of the edge $e_j \in P$, for $j=1, \ldots, k$, with the vertices of $Q$.  Clearly, $V_{Q, e_j}$ is a collection of $n$ parallel segments.

    A face of the complement of $P \oplus Q$ is identical to a face of the overlay of
    \begin{equation*}
        V_{Q, p_1}, \ldots, V_{Q, p_k}, V_{Q,e_1}, \ldots, V_{Q, e_k}.
    \end{equation*}
    The total complexity of these $2k$ arrangements is $O(nk)$. Thus, by \thmref{complex:Face:In:Overlay}, the complexity of such a face is
    \[
        O \pth{ nk \frac{\lambda_{3}(k)}{k} } = O \pth{ nk \alpha(k) }.
    \]
\end{proof}

\begin{remark}
    See \cite{HCAHS95} for an alternative proof of \thmref{mink} and for a lower bound construction which shows that this bound is tight in the worst case.
\end{remark}

An algorithmic consequence of \thmref{complex:Face:In:Overlay} is:

\begin{lemma}
    Given $t$ arrangements $A(\Gamma_{1}), \ldots, A(\Gamma_{t})$, and a point $p$, satisfying the assumptions made above, we can calculate the face containing $p$ in the overlay arrangement $A( \Gamma)$ in (deterministic) time
    \[ O \pth{ \frac{ \lambda_{s+2}(t) } { t } C \log{ \pth{ \frac{ \lambda_{s+2}(t) } { t } C } } \log{t} },
    \]
    where $C$ is the total complexity of all the marked faces (containing $p$) in the given arrangements.  \lemlab{stupid}
\end{lemma}

\begin{proof}
    We use the algorithm described in \cite{GSS89} which calculates the face containing the point $p$ in the overlay of two arrangements, and we merge repeatedly two arrangements at a time in a balanced binary-tree fashion.  By \thmref{complex:Face:In:Overlay}, the total complexity of all the marked faces in each level of the tree is at most $O( \frac{\lambda_{s+2}(t)}{t} C )$, and the tree depth is $\lceil \log{t} \rceil$. The result then follows from the analysis of \cite{GSS89}, since the cost of a single application of the merging algorithm of \cite{GSS89} is $O \pth{ k \log{k} }$, where $k$ is the overall complexity of the input and output faces participating in the merge.
\end{proof}

\blankline

This algorithm slightly improves the best known solutions for the problem of translational motion planning of a simple polygon $P$ with $k$ sides amidst a collection of $n$ point-obstacles (imagine that $P$ is translating on a board amidst a collection of pins or pegs tacked to the board). Let $s, e$ be the given starting and finish points, respectively. Let $c$ be an arbitrary point inside $P$.

A solution to the above motion planning problem is a path from $s$ to $e$, such that $P$ should not intersect any obstacle when its center is located on any point of the path. To do so, we reduce the problem into the calculation of a regular path between $e$ and $s$. We notice that each obstacle point $p_i$ induces a forbidden region that the path should avoid, namely the polygon $P$ rotated by $180$ degrees, and centered at $p_i$.  We replace each obstacle point by its forbidden region, $R_i$.  Let $\R = \{ R_1, \ldots, R_n \}$. Any path in the arrangement $FP = A(\R)$ between the points $s$ and $e$, which is contained in the complement of the union of the $R_i$'s yields an obstacle-avoiding motion of the required kind. Such a path exists iff $e$ and $s$ are contained in the same face of $A(\R)$.  We calculate the face $F$ containing $e$ in the arrangement $FP$.  Now, in linear time in the size of $F$, we can check if $s$ is contained in $F$ and if so to calculate the required path.

For the given problem, we can calculate a connected component (a face) of $F P$ in time
\[
    O( n k \alpha(k) \log{( n k \alpha(k) )} \log{k} ) = O( n k \alpha(k) \log{n k} \log{k} ).
\]

Indeed, all the expanded obstacles are translated copies of the same polygon, and this polygon has $k$ sides: $s_1, \ldots, s_k$. Let $T_i$ be the set of the $n$ translated copies of $s_i$ appearing in $R_1, \ldots, R_n$, for $i=1, \ldots, k$. Let $T = \cup_{i=1}^{k} T_i$. Clearly, $A(T) = FP$.

So we calculate each $A(T_i)$ separately (in $O(n\log{n} )$ time, since this arrangement is made of $n$ translated copies of the same segment), and apply the algorithm described in \lemref{stupid} to calculate the face containing $s$ in $FP = A(T)$.

This slightly improves the bound $O( n k \alpha(n k) \log{n k} \log{n} )$, given in \cite{GSS89}, assuming $k \ll n$.

\subsection{The Complexity of a Face %
   in Certain Arrangements of Line Segments}
\seclab{line:seg:arrangement}

In this section, we present new bounds, for some special cases, on the complexity of a single face in an arrangement of line segments.  The following lemma is the main tool in proving these new results.

\begin{lemma}
    Let $S = \{ s_1, \ldots, s_n \}$ be a collection of $n$ line segments.  Let $\gamma$ be a Jordan arc in the plane intersecting exactly once each segment of $S$. Then the complexity of a face in the arrangement $A = A(S \cup \{ \gamma \} )$ is $O(n)$.

    \lemlab{arc:and:segments}
\end{lemma}

\begin{proof}
    Without loss of generality, we assume that the face of interest is the unbounded face of $A$. We split each segment $s_i$, for $i=1, \ldots, n,$ into two line segments, $t_i,r_i$, at its intersection point with $\gamma$.  Let $S_1 = \{ t_1, \ldots, t_n \}$ be the set of segments lying on one side of $\gamma$ and let $S_2 = \{r_1, \ldots, r_n \}$ be the set of segments lying on the other side.

    We derive a linear bound on the complexity of the unbounded face of $A(S_1 \cup \{ \gamma \})$, and symmetrically for $A \pth{ S_2 \cup \{ \gamma \} }$, which implies the asserted bound, by applying the single face combination \thmref{complex:Face:In:Overlay} to $A( S_1 \cup \{\gamma \} )$ and $A( S_2 \cup \{ \gamma \} )$.

    Let $p_i$ and $q_i$ be the endpoints of $t_i$, for $i =1, \ldots, n$, where $q_i \in \gamma$.  We orient $t_i$ from $q_i$ to $p_i$.

    Let $C$ be the boundary of the unbounded face of $A \pth{ S_1 \cup \{ \gamma \} }$. Clearly, $C$ is connected. Let $v \in C$ be an arbitrary point. As in the proof of \thmref{single_Face_Complex}, we split each segment of $S_1$ into at most four symbols, such that the sequence $\sigma$ of symbols encountered as we trace $C$ from $v$ in a counterclockwise direction, is a $DS(4n, 3)$ sequence.

    We partition $S_1$ into the set $S_{l}$ of left-side symbols (under the orientation assumed above), and the set $S_{r}$ of right-side symbols.

    We show that the restricted sequences $\sigma_{S_{r}}, \sigma_{S_{l}}$ are $DS(2n, 2)$ sequences, implying that their lengths are linear. The combination \thmref{restrict:DS} for \DS sequences then implies that the length of $S$ is linear.

    Suppose, for the sake of contradiction, that there exists a subsequence of the form $\langle p \cdots q \cdots p \cdots q \rangle$ in $\sigma_{S_{l}}$. Let $t_p, t_q$ denote the segments of $S_1$ containing $p,q$, respectively.  Using the Consistency Lemma, and reasoning as in the proof of \thmref{single_Face_Complex}, one can show that $t_p$ and $t_q$ must intersect, and that the two appearances of $t_p$ in the subsequence correspond to two subintervals lying on different sides of the intersection point, and similarly for $t_q$. This however is impossible, because the left side of the second appearance of $t_q$ must then be trapped in the internal face created by $t_p,t_q, \gamma$; see \figref{trap}.

    Thus $\sigma_{S_{l}}$ is a $DS(2n, 2)$ sequence, and the same holds for $\sigma_{S_r}$. This completes the proof of the lemma.
\end{proof}

\begin{figure}[ht]
    \centering
    \includegraphics{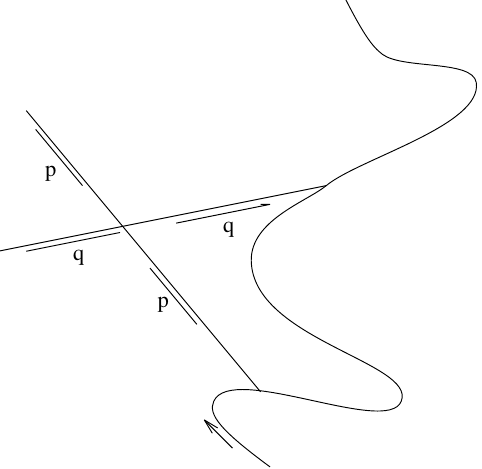}
    \caption{The subsequence $\cdots p \cdots q \cdots p \cdots q \cdots$ is impossible}
    \figlab{trap}
\end{figure}

\begin{remark}
    \lemref{arc:and:segments} is also implied by a similar observation in \cite{AHKMN94}.
\end{remark}

\begin{remark}
    \lemref{arc:and:segments} also holds when $\gamma$ intersects each segment of $S$ at least once and at most some constant number of times.
\end{remark}

\begin{remark}
    \remlab{seg:arr:intersect:seg}%
    Given a set $S = \{ s_1, \ldots, s_n \}$ of segments, all of which intersect a segment $s_k$, for some fixed $1 \leq k \leq n$. \lemref{arc:and:segments} implies that the complexity of a face in $A(S)$ is linear. This is somewhat surprising, considering that the complexity of a face in the arrangement $A(S \setminus \{ s_k \} )$ can be $\Omega( n \alpha(n ) )$, as follows from the construction of \cite{WS88}.
\end{remark}

\begin{defn}
    Given a set $S$ of $n$ Jordan arcs, we define the {\em covering number} of $S$ to be the minimal size of a subset $T \subseteq S$, such that each arc of $S$ intersects at least one arc of $T$.
\end{defn}

\begin{lemma}
    Given a set $S$ of $n$ line segments, with a covering number $k$.  The complexity of any face in the arrangement $A(S)$ is $O( n \alpha(k) )$.
\end{lemma}

\begin{proof}
    Let $p$ be a marking point inside a face $f$ of $A(S)$.  Let $T = \{ t_1, \ldots, t_k \}$ be a subset of $k$ segments realizing the covering number of $S$. We define $S_i^{\prime}$ to be the set of segments of $S$ intersecting $t_i$, for $i=1, \ldots, k$. Let
    \[
        S_i = S_i^{\prime} \setminus \pth{ \bigcup_{j=1}^{i-1} S_j^{\prime} \cup \bigcup_{j=i+1}^{k} \{ t_j \} },
    \]
    for $i=1, \ldots, k$.  Thus $t_i \in S_i$ and $\cup_{i=1}^{k} S_i = S$, for $i=1, \ldots, k$.  Also $S_i \cap S_j = \emptyset$, for $1 \leq i < j \leq k$.  By \remref{seg:arr:intersect:seg} it follows that the complexity of the face containing $p$ in the arrangement $A(S_i)$, for $i=1, \ldots, k$, is $O( |S_i| )$.

    Applying the single face Combination \thmref{complex:Face:In:Overlay} to $A(S_1), \ldots, A(S_k)$, we conclude that the complexity of $f$ is $O\pth{ n\alpha(k) }$, as claimed.
\end{proof}

\begin{lemma}
    Let $\brc{ \gamma_1, \ldots, \gamma_k }$ be a collection of $k$ Jordan arcs, disjoint except at their endpoints, such that $\gamma = \cup_{i=1}^{k} \gamma_i$ is a simple close Jordan curve. Let $S = \{ s_1, \ldots, s_n \}$ be a collection of $n$ line segments, such that the following holds:
    \begin{itemize}
        \item The relative interior of $s_i$ is contained in the interior of $\gamma$.

        \item The endpoints of $s_i$ lie on $\gamma$.

        \item The endpoints of any segment $s_i = p_i q_i$ lie in the relative interiors of two distinct arcs $\gamma_{l_i}$ and $\gamma_{k_i}$.
    \end{itemize}
    Then the complexity of a face in the arrangement $A(S)$ is $O(n \alpha(k ))$.  \lemlab{closed:arc:segments}
\end{lemma}

\begin{proof}
    Let $p$ be a marking point inside a face $f$ of $A(S)$.  Let $S_{i j}$, for $i \ne j$, denote all the segments of $S$ having one endpoint in $\gamma_i$ and the other endpoint in $\gamma_j$. Clearly $\cup_{1 \leq i < j \leq k} S_{i j} = S$ and these sets are pairwise disjoint.  We claim that the complexity of the face containing $p$ in the arrangement $A(S_{i j})$ is $O(|S_{i j}|)$.

    There are two cases. If $p$ lies in an internal face of $A(S_{i j})$, then $f$ is also the faces containing $p$ in the arrangement $A(S_{i j} \cup \{ \gamma_i \} )$.  By \lemref{arc:and:segments} the complexity of this face is $O( |S_{i j}| )$.

    If $p$ lies in the unbounded face of $A(S_{i j})$, then it is easily seen that each vertex (an intersection point of two segments or an endpoint of a segment) is either a vertex of the unbounded face of $A_1 = A(S_{i j} \cup \{ \gamma_i \} )$ or of $A_2 = A(S_{i j} \cup \{ \gamma_j \})$.  By \lemref{arc:and:segments} the complexity of the unbounded faces of $A_1$ and of $A_2$ is $O( |S_{i j}| )$, implying that the complexity of the unbounded face of $A(S_{i j})$ is $O(|S_{i j}| )$.

    Applying \thmref{complex:Face:In:Overlay} to the arrangements $A(S_{i j})$, for $1 \leq i < j \leq k$, we conclude that the complexity of $f$ is
    \[
        O \pth{ \pth{\sum_{1 \leq i < j \leq k} |S_{i j}| } \alpha(k^{2}) } = O \pth{ n \alpha(k) },
    \]
    since $\alpha( k^{2} ) = O( \alpha(k) )$.
\end{proof}

\begin{lemma}
    Let $C$ be a circle of radius $r$ and let $S = \{ s_1, \ldots, s_n \}$ be a collection of $n$ chords of $C$, such that $length(s_i) \geq c r$, for some constant $c$ and for $i=1, \ldots, n$. Then the complexity of a face in the arrangement of $A(S)$ is $O(n)$, where the constant of proportionality depends on $c$.  \lemlab{circle:chords}
\end{lemma}

\begin{proof}
    We break $C$ into $k = 4\ceil{ \frac{\pi}{c} }$ arcs, $C_1, \ldots, C_k$, each of length $\frac{2\pi r}{k} \approx \frac{c r}{2}$. This implies that no segment of $S$ has both endpoints on the same arc $C_i$, for any $1 \leq i \leq k$.  By \lemref{closed:arc:segments} the complexity of any face of $A(S)$ is $O(n \alpha(k)) = O(n)$, since $k$ is a constant.
\end{proof}

\begin{remark}
    Clearly, \lemref{circle:chords} can be extended to other curves and appropriate collections of chords.
\end{remark}

\begin{lemma}
    Given a set $S = \{ s_1, \ldots, s_n \}$ of $n$ line segments, and a set \\ $L = \{ l_1, \ldots, l_k \}$ of $k$ lines, such that the endpoints of each $s_i$ lie on $\cup_{i=1}^{k} l_i$, then the complexity of a face of $A(S)$ is $\Theta(n \alpha(k) )$.
\end{lemma}

\begin{proof}
    The construction described in \secref{face:introduction} implies the lower bound.

    The upper bound proof is similar to the proof of \lemref{closed:arc:segments}.

    Let $p$ be a marking point inside the given face $F$ of $A(S)$. Let $S_{i j}$, for $1 \leq i \leq j \leq k$, denote the set of all the line segments of $S$ having one endpoint on $l_i$ and the other on $l_j$.  Assume, without loss of generality, that these sets are disjoint.

    Clearly, the complexity of the face containing $p$ in the arrangement $A(S_{i{} i})$, for any $i=1,\ldots, k$, is $O(|S_{i{} i}|)$, since this arrangement is a union of segments lying on a common line.

    We next show that the complexity of the face containing $p$ in $A(S_{i j})$, for $1 \leq i < j \leq k$, is linear. We partition $S_{i j}$ into $4$ subsets of line segments, and prove the linear bound for each subset separately. The combination \thmref{complex:Face:In:Overlay}, applied to these four arrangements, implies the asserted linear bound.

    The lines $l_i, l_j$ divide the plane into $4$ quadrants. We consider separately the line segments of $S_{i j}$ contained in each quadrant. Let $S_1, S_2, S_3, S_4$ denote the resulting partition of $S_{i j}$.

    Let $s$ (respectively, $t$) denote the smallest line segment contained in $l_i$ (respectively, $l_j$) that intersects all the segments of $S_1$. Arguing as in the proof of \lemref{closed:arc:segments}, we conclude that the complexity of the face of $A(S_1)$ that contains $p$ is $O(|S_1|)$.

    In a similar manner we prove that the complexities of the face containing $p$ in the arrangements $A(S_2), A(S_3), A(S_4)$ are also linear. \thmref{complex:Face:In:Overlay} then implies that the complexity of the face containing $p$ in the arrangement $A(S_{i j})$ is $O(|S_{i j}|)$.

    Applying \thmref{complex:Face:In:Overlay} once again to the arrangements $A(S_{i j})$, for $1 \leq i \leq j \leq k$, it follows that the complexity of $F$ is
    \[
        O \pth{ (\sum_{1 \leq i \leq j \leq k} |S_{i j}| ) \alpha( k^{2} + k ) } = O \pth{ n \alpha(k ) },
    \]
    as claimed.
\end{proof}

\section{Complexity of Many Faces in Arrangements}

\subsection{Many Faces in the Overlay of Many %
   Arrangements}
\seclab{many:faces:overlay}

Given $k$ points $p_{1}, \ldots, p_{k}$ and $t$ arrangements $A^{1} = A(\Gamma_{1}), \ldots, A^{t}=A(\Gamma_{t})$, let $A$ denote, as above, the arrangement resulting from the overlay of these $t$ arrangements. For each point $p_{i}$, for $1 \leq i \leq t$, let $E_{i}$ denote the face of $A$ containing $p_{i}$.

Let $F$ be a marked face (containing one of the points $p_j$) in $A^{i}$, for $1 \leq i \leq t$, and let $\zeta$ be a connected component of the boundary of $F$.

We trace $\zeta$ in a counterclockwise direction, and split the arcs of $\Gamma_i$ that we encounter into subarcs at points where one of the following events happens:

\begin{itemize}
    \item We enter the boundary of a new face $E_{i}$ (i.e., a face different from the {\em last} face $E_j$ that we have encountered).

    \item We reach a vertex of $A^{i}$ lying on $\zeta$ (i.e., an intersection point of two distinct arcs of $\Gamma_i$).
\end{itemize}

Let $G_{\zeta}$ denote the set of the resulting split subarcs of $\zeta$.  The set $G_{\zeta}$ is called a {\em refinement} of $\zeta$.  We now replace all the arcs appearing along $\zeta$, in $A^{i}$, by the arcs of $G_{\zeta}$. We perform this replacement process over all the boundary components of all the marked faces in $A^{1}, \ldots, A^{t}$.  Let $G_{1}, \ldots, G_{t}$ be the sets of new arcs generated in this manner for $A^{1}, \ldots, A^{t}$, respectively. We call the arrangement $A(G_{i})$, for $1 \leq i \leq t$, the {\em splitting arrangement} of $A^{i}$ and the set $G_{i}$ the {\em refinement} of $\Gamma_i$.

Having constructed the splitting arrangements $A(G_{i})$, we can replace each $A^{i}$ by $A(G_{i})$ in the overlay arrangement $A$ without reducing the combinatorial complexity of the marked faces in $A$.  The complexity of the marked faces in $A(G_{i})$ is at most $2|G_{i}|$ (a subarc may appear on the boundary of two marked faces, one on each side). Let $L_{i}$ be the number of splitting points created for $G_{i}$.  We denote by $L = \sum_{i=1}^{t} L_{i}$ the {\em splitting number} of $A$. It follows immediately from the construction that:

\begin{corollary}
    \corlab{input:complexity}%
    Let $C_{i}$ be the total complexity of the marked faces in $A^{i}$.  Then $|G_{i}| \leq 2L_{i} + 2C_{i}$ and the total complexity of the marked faces in $A(G_{1}), \ldots, A(G_{t})$ is at most $2L + 2C$.
\end{corollary}

\begin{lemma}{\bf{(General Combination Lemma for Arrangements of Curves)}}{}
    \lemlab{alt:Many:Arrange:Many:Faces} \\ Given $k$ points, $p_{1}, \ldots, p_{k}$, and $t$ arrangements, $A^{1}, \ldots, A^{t}$, let $A$ denote the arrangement resulting from the overlay of $A^{1}, \ldots, A^{t}$. For each point $p_{i}$, for $1 \leq i \leq t$, let $E_{i}$ denote the face of $A$ containing $p_{i}$, and let $C$ denote the total complexity of the faces of each of the $t$ arrangements, which contain at least one of the given points.  Let $L$ denote the splitting number of $A$.  If every pair of arcs in the given arrangements intersect at most $s$ times, then the total complexity of all the faces $E_{i}$, for $1 \leq i \leq k$, is $O( \frac{\lambda_{s+2}(t)}{t} ( C + L ))$.
\end{lemma}

\begin{proof}
    Let $G_{1}, \ldots, G_{t}$ be the refinements of $\Gamma_{1}, \ldots, \Gamma_{t}$, respectively.  Let $G = \cup_{i=1}^{t} G_{i}$ be the union of the given refinements, and let $A_{G} = A( G )$ be the overlay arrangement of the $t$ arrangements $A(G_i)$, for $1 \leq i \leq t$. As noted above, the total complexity of the marked faces in $A_{G}$ is no less than the complexity of the marked faces in $A$.

    Since each arc of $G_{i}$, for $1 \leq i \leq t$, appears in the boundary of at most two marked faces (each side of the arc can appear on the boundary of at most a single face) in $A_{G}$, the result follows from \thmref{complex:Face:In:Overlay}.  Formally, let $V^{r}_{i}$, for $1 \leq r \leq t$, denote the set of arcs of $G_{i}$ that appear along the boundary of $E^{G}_{i}$ ($E^{G}_i$ is the face in $A_{G}$ that contains the point $p_{i}$).  $V^{r}_{i}$ contain edges that do not intersect in their interior. Thus the total complexity of the arrangement defined by the arcs in $V^{r}_{i}$ is $O(|V^{r}_{i}|)$.

    The face $E^{G}_{i}$ is the face, containing $p_{i}$, in the arrangement created from the overlay of the arrangements $A(V^{1}_{i}), \ldots, A(V^{t}_{i})$. Using \thmref{complex:Face:In:Overlay}, it follows that the total complexity of $E^{G}_{i}$ is
    \[
        O \pth{ \frac{\lambda_{s+2}(t)}{t} \sum_{r=1}^{t} |V^{r}_{i}| }.
    \]
    Thus the total complexity of all the faces $E^{G}_{i}$, for $1 \leq i \leq k$, is
    \[ O \left( \sum_{i=1}^{k} \left( \frac{\lambda_{s+2}(t)}{t} \sum_{r=1}^{t} |V^{r}_{i}| \right) \right) = O \left( \frac{\lambda_{s+2}(t)}{t} \sum_{i=1}^{k} \sum_{r=1}^{t} |V^{r}_{i}| \right) = O \left( \frac{\lambda_{s+2}(t)}{t} \sum_{r=1}^{t} \sum_{i=1}^{k} |V^{r}_{i}| \right).
    \]

    Since $\sum_{i=1}^{k} |V^{r}_{i}| \leq 2|G_{i}|$, the above expression is
    \[ O \left( \frac{\lambda_{s+2}(t)}{t} |G| \right) = O \left( \frac{\lambda_{s+2}(t)}{t} ( L + C ) \right).
    \]
\end{proof}

The following lemma bounds the splitting number $L$, using a planarity argument.

\begin{lemma}
    \lemlab{splitting:number:bound} The splitting number $L$ of $A$ is at most $2k t + 2C$, where $C$ is the total complexity of the marked faces in the $t$ separate arrangements, and $k$ is the number of marking points.
\end{lemma}

\begin{proof}
    Fix $1 \leq r \leq t$, let $G_{r}$ be the refinement of $A^{r}$, and let $Z_{r}$ denote the sets of marked faces of $A^{r}$. For a marked face $F$ of $A^r$, we let $k(F)$ denote the number of the marking points contained in $F$, let $C_{F}$ denote the complexity of $F$, and let $l(F)$ denote the number of boundary components of $F$.

    Fix a face $F \in Z_{r}$ and let the $k(F)$ points contained in $F$ be $p_{1}, \ldots, p_{k(F)}$.  Let $\zeta_{1}, \ldots, \zeta_{l(F)}$ be the distinct connected components of $\partial{F}$. For each of these points $p_{i}$ let $E_{i}$ denote the face in the arrangement $A$ that contains the point $p_{i}$. Traverse each $\zeta_{m}$ and partition it into connected portions $\delta$ so that each such portion intersects the boundary of only a single region $E_{i}$ (and so that two adjacent portions intersect distinct such regions); note that in general the endpoints of the portions $\delta$ are not uniquely defined. (This partitioning differs from the refinement $G_r$ in that here we do not split arcs at vertices of $A^{r}$.) We define, in a manner similar to that in \cite{GSS89}, a plane embedding of a planar bipartite graph $K$ as follows.  The vertices of $K$ are the points $p_{1}, \ldots, p_{k(F)}$ and $l(F)$ additional points $q_{1}, \ldots, q_{l(F)}$, so that $q_{i}$ lies inside the connected component $H_{i}$ of $\Re^{2} - F$ whose common boundary with $F$ is $\zeta_{i}$. For each subarc $\delta$ of $G_r$ lying on some $\zeta_{m}$ and intersecting some $\partial{E_{i}}$, we add the edge $(q_{m}, p_{i})$ to $K$, and draw it by taking an arbitrary point in $\delta \cap \partial{E_{i}}$, and connect it to $p_{i}$ within $E_{i}$ and to $q_{m}$ within $H_{m}$.  The connectedness of each $E_{i}$ and each $H_{m}$ implies, as in \cite{GSS89}, that we can draw all edges of $K$ so that they do not cross one another.  It follows from the definition of the portions $\delta$ that in this embedding of $K$ each face is bounded by at least four edges (even though $K$ may have multiple edges between a pair of vertices).  Thus, by {{Euler's}} formula, and because $K$ bipartite, the number of edges in $K$, and thus the number of portions $\delta$, is at most $2(k(F) + l(F))$.

    Applying this process to all marked faces in $A^{r}$, it follows that the splitting number of $A^{r}$ (which is at most equal to the number of portions $\delta$ created plus the number of vertices of the marked faces in $A^{r}$) is $2k + 2C_r$ where $C_r$ is the total complexity of the marked faces in $A^{r}$.  Summing this quantity over all the $t$ arrangements, the lemma follows.
\end{proof}

Substituting the bound of \lemref{splitting:number:bound} into \lemref{alt:Many:Arrange:Many:Faces}, we obtain the following more explicit version of \lemref{alt:Many:Arrange:Many:Faces}:

\begin{theorem}{\bf (General Combination Theorem for Arrangements of Curves)}
    \thmlab{many:Arrange:Many:Faces}%

    Given $k$ points $p_{1}, \ldots, p_{k}$ and $t$ arrangements $A^{1}, \ldots, A^{t}$, let $A$ denote the arrangement resulting from the overlay of these $t$ arrangements. For each point $p_{i}$, for $1 \leq i \leq t$, let $E_{i}$ denote the face of $A$ containing $p_{i}$, and let $C$ denote the total complexity of the faces containing at least one of the given points in each of the $t$ arrangements. If every pair of arcs in the given arrangements intersect at most $s$ times, then the total complexity of all the faces $E_{i}$, for $1 \leq i \leq k$, is $O \pth{ k \lambda_{s+2}(t) + \frac{\lambda_{s+2}(t)}{t} C }$.
\end{theorem}

\begin{proof}
    Obvious.
\end{proof}

A simple but useful application of \thmref{many:Arrange:Many:Faces} is:

\begin{theorem}
    \thmlab{m:Faces:In:Arrangement} Given a collection $\Gamma = \{ \gamma_{1}, \ldots, \gamma_{n} \}$ of $n$ Jordan arcs, each pair of which have at most $s$ intersection points, the complexity of any $m$ faces in the arrangement $A(\Gamma)$ is at most $O \pth{ \sqrt{m} \lambda_{s+2}( n ) }$.
\end{theorem}

\begin{proof}
    Let $P$ be a set of $m$ points, one inside each of the given faces.  Divide $\Gamma$ into $t = \ceil{ \frac{n}{\sqrt{m}} }$ sets, $\Gamma_{1}, \ldots, \Gamma_{t}$, so that each set contains at most $\sqrt{m}$ arcs of $\Gamma$. The complexity of the marked faces (by points from $P$) in any $A(\Gamma_{i})$, for $1 \leq i \leq t$, is at most $O(m)$ (this is the total complexity of the arrangement).  By \thmref{many:Arrange:Many:Faces}, the complexity of the marked faces in $A(\Gamma)$, regarded as the overlay of the arrangements $A(\Gamma_{1}), \ldots, A(\Gamma_{t})$, is
    \[ O \left( \frac{\lambda_{s+2}(t)}{t} m t + m \lambda_{s+2}(t) \right) = O \left( m \lambda_{s+2}(t) \right) = O \left( \lambda_{s+2} \left( \frac{n}{\sqrt{m}} \right) m \right)
    \]
    Since $\lambda_{s+2} \left( \frac{n}{\sqrt{m}} \right) \leq \frac{ \lambda_{s+2}(n) }{ \sqrt{m} }$, it follows that the complexity of the $m$ given faces is
    \[
        O \pth{ \frac{\lambda_{s+2}( n ) } {\sqrt{m}} m } = O \pth{ {\sqrt{m}} \lambda_{s+2}( n ) },
    \] as claimed.
\end{proof}

The bound in \thmref{m:Faces:In:Arrangement} has already been proved in \cite{EGPPSS92}, but the proof given here is much simpler than the previous proof.  The theorem provides the best currently known bound for the complexity of many faces in arrangements of general arcs. We note, however, that for the cases of lines, segments, or circles, better bounds are known (See \cite{cegsw-ccbac-90, EGS90}). %
For example, the complexity of $m$ faces in an arrangement of $n$ line segments is $O \pth{ m^{2/3} n^{2/3} + n \alpha(n) + n \log{m} }$; see \cite{cegsw-ccbac-90, AEGS92, EGS90}.

We also note that the approximation $\lambda_{s+2} \pth{ \frac{n}{\sqrt{m}} } \leq \frac{\lambda_{s+2}(n)}{\sqrt{m}}$ is fairly sharp only when $m \ll n^{2}$. When $m$ approaches its maximum value $\Theta( n^{2} )$, the non-approximated bound $O \pth{ \lambda_{s+2} \pth{ \frac{n}{\sqrt{m}} } m }$ approaches $O(m)$, which is the correct bound. This does not hold in the older proof of \cite{EGPPSS92}.

To demonstrate that \thmref{many:Arrange:Many:Faces} is almost tight in the worst case, we need the following lemma (See \cite[pp. 111--112]{e-acg-87}) :

\begin{lemma}
    The maximum complexity of $m$ faces in an arrangement of $n$ lines or line segments is $\Omega \pth{ m^{2/3}n^{2/3} +n }$.  \lemlab{lower:bound:many:faces}
\end{lemma}

We can restate the construction given in \lemref{lower:bound:many:faces} in terms of overlaying arrangements.  We partition the constructed set of lines into maximal sets, $\Gamma_1, \ldots, \Gamma_t$, each consisting of parallel lines, where
\[
    t = 2 + \sum_{s=1}^{f(n,m)} \phi(s) = O \left( \frac{n^{{2/3}}}{m^{{1/3}}} \right).
\]
Let $P$ be the set of $m$ marking points in $A = A(G)$. Clearly, the overall complexity of the marked faces in $A(\Gamma_1), \ldots, A(\Gamma_t)$ is $O( n)$ (this is the overall complexity of these entire arrangements).  Applying \thmref{many:Arrange:Many:Faces} in this setup, we conclude that the total complexity of the marked faces in $A$ is
\[
    O \left( \frac{\lambda_{3}( t )}{t} n + m \lambda_3(t) \right) = O \left( \left( n + m t \right)\alpha(t) \right) = O \left( \left( n + n^{{2/3}} m^{{2/3}} \right) \alpha(n) \right).
\]

Comparing this with \lemref{lower:bound:many:faces}, we can conclude that the bound of \thmref{many:Arrange:Many:Faces} is very close to tight (within a factor of $\alpha(n)$) in the worst case.

Nevertheless, we believe that the term $O( k \lambda_{s+2}(t) )$ in the bound of \thmref{many:Arrange:Many:Faces} is too large in many specific applications, and an interesting direction for further research is to improve the analysis of \lemref{splitting:number:bound}, so as to get better bounds on the splitting number in various specific cases.

\subsection{The Complexity of Many Faces %
   in Sparse Arrangements}
\seclab{sparse:arrangement}

In this section, we derive some bounds on the complexity of many faces in {\em sparse} arrangements, namely arrangements whose complexity is much smaller than quadratic (in the number of arcs).

\begin{defn}
   Let $\Gamma = \{ \gamma_1, \ldots, \gamma_n \}$ be a set of $n$ Jordan arcs in the plane, so that each pair of arcs from $\Gamma$ have at most $s$ intersection points. We call a function $f:\Gamma \rightarrow \{ 1, \ldots, t \}$ a {\em coloring} of $\Gamma$ by $t$ colors, if $f(\gamma_i) \ne f(\gamma_j)$ for each intersecting pair of arcs $\gamma_i, \gamma_j \in \Gamma$.  The smallest number $t$, for which there exists a coloring of $\Gamma$ with $t$ colors, is called the {\em chromatic number} of $\Gamma$, and is denoted by $\chi( \Gamma )$.

   Let $G(\Gamma) = (\Gamma, E)$ be the graph, such that $(\gamma_i,\gamma_j) \in E$ {iff} $\gamma_i \cap \gamma_j \ne \emptyset$. We call the graph $G(\Gamma)$ the {\em induced graph} of $\Gamma$. Clearly, the chromatic number $\chi(G(\Gamma))$ of $G(\Gamma)$ is equal to $\chi(\Gamma)$, and the complexity of $A(\Gamma)$ is $\Theta( |\Gamma| + |E| )$ (where the constant of proportionality is linear in $s$).
\end{defn}

The following result shows that we can bound the complexity of $k$ faces in an arrangement $A(\Gamma)$, in terms of the chromatic number of $\Gamma$:

\begin{lemma}
    \lemlab{chi:many:faces} Given a set $\Gamma$ of $n$ arcs, as above, with $t=\chi(\Gamma)$, then the complexity of any $k$ faces in $A(\Gamma)$ is
    \[
        O \left( \frac{\lambda_{s+2}(t)}{t} n + k \lambda_{s+2}(t) \right).
    \]
\end{lemma}

\begin{proof}
    Let $f:\Gamma \rightarrow \{1, \ldots, t \}$ be a coloring of $\Gamma$ by $t$ colors. Let $\Gamma_i = f^{-1}(i)$, for $i = 1, \ldots, t$. By definition, no pair of arcs from the same $\Gamma_i$ intersect. Thus the complexity of $A(\Gamma_i)$ is $O( |\Gamma_i| )$. Applying \thmref{many:Arrange:Many:Faces} to the overlay of $A(\Gamma_1), \ldots, A(\Gamma_t)$, it follows that the complexity of any $k$ faces in $A(\Gamma)$ is
    \[
        O \left( \frac{\lambda_{s+2}(t)}{t} n + k \lambda_{s+2}(t) \right).
    \]
\end{proof}

\begin{remark}
    Again, the bound in \lemref{chi:many:faces} is almost tight in the worst case. Indeed, the chromatic number of the arrangement $A$ constructed in \lemref{lower:bound:many:faces} is exactly the number of subarrangements used. Thus, \lemref{chi:many:faces} implies a bound of $O \pth{ \pth{ n + n^{{2/3}} k^{{2/3}} } \alpha(n) }$ on the complexity of $k$ faces in $A$.  Comparing this with \lemref{lower:bound:many:faces}, we can conclude that the bound of \lemref{chi:many:faces} is very close to tight (within a factor of $\alpha(n)$) in the worst case.
\end{remark}

Since $\chi(G)$ is at most the maximal degree of a vertex of $G$ plus one, we obtain:

\begin{lemma}
    \lemlab{limited:intersect:many:faces} Given a set $\Gamma$ of $n$ arcs as above, such that any arc from $\Gamma$ intersects at most $t$ others arcs from $\Gamma$. Then, the complexity of any $k$ faces in $A(\Gamma)$ is
    \[
        O \left( \frac{\lambda_{s+2}(t)}{t} n + k \lambda_{s+2}(t) \right).
    \]
\end{lemma}

If $A(\Gamma)$ is a sparse arrangement, we can obtain the following improved bound on the complexity of many faces in it, using the following simple observation (see \cite{a-hwcn-85}):

\begin{lemma}
    The chromatic number $\chi(G)$ of a graph $G = (V, E)$ satisfies\\
    $\chi(G) = O \pth{ \sqrt{|E|} }$.  \lemlab{chromatic:bound}
\end{lemma}

As an application of \lemref{chromatic:bound}, we obtain:

\begin{lemma}
    \lemlab{bad:sparse:many:faces} Given a set $\Gamma$ of $n$ arcs as above, such that the total complexity of $A(\Gamma)$ is $w$, then the complexity of any $m$ faces in $A(\Gamma)$ is
    \[ O \left( \frac{\lambda_{s+2}(\sqrt{w})}{\sqrt{w}} n + m \lambda_{s+2}(\sqrt{w}) \right).
    \]
\end{lemma}

\begin{proof}
    By \lemref{chromatic:bound}, we have $\chi(\Gamma) = O \pth{ \sqrt{w} }$, so the bound is an immediate consequence of \lemref{chi:many:faces}.
\end{proof}

Note that the bound of \lemref{bad:sparse:many:faces} is rather weak. For example, when $w= O(n^2)$, we only get a bound of $O \pth{m \lambda_{s+2}(n) }$, which is much weaker than the bound given in \thmref{m:Faces:In:Arrangement}. We now improve the bound of \lemref{bad:sparse:many:faces}, using a more careful and more involved analysis.

\begin{theorem}
    \thmlab{sparse:m:faces} Let $\Gamma = \{ \gamma_1, \ldots, \gamma_n \}$ be a set of $n$ Jordan arcs in the plane, such that each arc of $\Gamma$ has at most $\nu$ points of local $x$-extremum, and each pair of arcs of $\Gamma$ have at most $s$ intersection points, for some constants $\nu$ and $s$.  Let $w$ be the total complexity of the arrangement $A(\Gamma)$. Then the complexity of any $m$ distinct faces in $A(\Gamma)$ is
    \[
        O \left( (n + \sqrt{m}\sqrt{w}) \frac{\lambda_{s+2}(n)}{n} \right)\,.
    \]
\end{theorem}

\begin{proof}
    The proof is similar to the proof of Theorem 3.2 in \cite{AEGS92}.

    Let $P$ be a set of $m$ marking points, one point in the interior of each of the $m$ faces of $A(\Gamma)$ under consideration.  Choose a random sample $R \subset \Gamma$ of size $r$, where $r$ is a parameter to be specified later. We decompose $A(R)$ into pseudo-trapezoidal subcells by drawing a vertical line segment upwards and downwards from each endpoint of an arc of $R$, each intersection point between arcs of $R$, and each locally $x$-extremal point on any edge of $A(R)$, and extend each of these segments until it hits another arc of $R$ or, failing this, all the way to $\pm \infty$.  Let $\tau$ denote the number of resulting pseudo-trapezoids.  Clearly, the expected value of $\tau$ is proportional to the total expected complexity of the arrangement $A(R)$.

    Since each intersection point of $A(R)$ is also an intersection point of $A( \Gamma)$, it follows that the expected size of $A(R)$ is proportional to the expected number of intersection points of $A(\Gamma)$ appearing in $A(R)$, plus the sample size $r$. Since each intersection point of $A(\Gamma)$ has probability $\frac{r(r-1)}{n(n-1)}$ to appear in $A(R)$, it follows that the expected number of intersection points in $A(R)$ is $O \pth{ \frac{r(r-1)}{n(n-1)}w }$. %
    Thus the expected complexity of $A(R)$ is $O \pth{ r + \frac{r(r-1)}{n(n-1)}{w} } = O \pth{ r + \frac{r^2}{n^2} w }$.

    Let $f_1, \ldots, f_{\tau}$ denote the resulting pseudo-trapezoids.  Let $n_i$ be the number of arcs of $\Gamma$ that intersect $f_i$, and let $m_i = |P \cap f_i|$ be the number of points contained in $f_i$. We can perturb the points, if necessary, so as to assume that none of them lies on any of the vertical segments added in the above decomposition. Hence we have $\sum_{i=1}^{\tau} m_i = m$.

    Let $A_i$ denote the subdivision of $f_i$ defined by the $n_i$ arcs intersecting $f_i$. We call a face of $A_i$ {\em coastal}, if it is incident to a vertical edge of the boundary of $f_i$. For each $i=1,\ldots,\tau$, we bound separately the complexity of the marked faces that are fully contained in $A_i$, and add up the resulting bounds. This may not yield an overall bound on the complexity of the marked faces, because (a) it ignores costal faces and (b) a coastal face may be contained in a marked face of $A(\Gamma)$, such that the marking point lies outside $f_i$.  To overcome these problems, we also add to our overall bound the total complexity of all coastal faces in each of the subarrangements $A_i$. It is easily seen that the resulting bound is indeed an upper bound on the complexity of the $m$ given faces.

    By \thmref{m:Faces:In:Arrangement}, the total complexity of the marked faces in $A_i$ is $O \left( \sqrt{m_i} \lambda_{s+2}(n_i) \right) \;$.  The complexity of the coastal faces in $A_i$ is $ O \left( \lambda_{s+2}(n_i) \right), $ by the Zone Theorem of \cite{EGPPSS92}.  Thus the total complexity of the $m$ marked faces in $A(\Gamma)$ is

    \begin{equation}
        \eqlab{bound}%
        O \pth{ \sum_{i=1}^{\tau} (1+\sqrt{m_i}) \lambda_{s+2}(n_i) } = O \pth{ \sum_{i=1}^{\tau} (1+\sqrt{m_i}) n_i } \cdot \frac{\lambda_{s+2}(n)}{n} \,.
    \end{equation}

    In the following analysis, we will use Theorem 3.6 of Clarkson and Shor \cite{cs-arscg-89}, which shows that the expected value of expressions of the form $\sum_{i=1}^{\tau} W \left( \binom{n_i}{d} \right)$, where $W$ is an arbitrary concave non-negative function and $d$ is a positive integer, satisfies
    \[
        E \pth{ \sum_{i=1}^{\tau} W \pth{ \binom{n_i}{d} } } \leq E( \tau ) \cdot W \pth{ D \pth{ \frac{n}{r} }^{d} },
    \]
    for some constant $D$.  We apply this theorem with $d=1$ and $W(x)\equiv x$, to obtain
    \[
        E \left( \sum_{i=1}^{\tau} n_i \right) \leq E(\tau) \cdot \frac{D n}{r},
    \]
    for some constant $D$. Hence the expected value of $\sum_{i=1}^{\tau} n_i$ is
    \[
        O \left( E( \tau) \cdot \frac{n}{r} \right) = O \left( ( r + \frac{r^2}{n^2} w ) \cdot \frac{n}{r} \right) = O \left( n + \frac{r}{n} w \right) \,.
    \]
    We next bound the expected value of the other sum, $\sum_{i=1}^{\tau} \sqrt{m_i} n_i$, appearing in \Eqref{bound}.  Using the Cauchy-Schwarz inequality, the expected value is

    \begin{eqnarray*}
      E \pth{ \sum_{i=1}^{\tau} \sqrt{m_i} n_i }
      &\le&
            E \pth{ \sqrt{
            \sum_{i=1}^{\tau} m_i } \sqrt{ \sum_{i=1}^{\tau} n_i^2
            } } = \sqrt{m} \cdot E \pth{ \sqrt{ \sum_{i=1}^{\tau}
            n_i^2 } }
            \leq
            \sqrt{m} \cdot \sqrt{ E \pth{
            \sum_{i=1}^{\tau} n_i^2 } },
    \end{eqnarray*}
    where the last inequality follows from the inequality $E \pth{ \sqrt{ X } } \leq \sqrt{E \pth{ X } }$.  Since $n_i^2 = 2 \choosex{n_i}{2} + n_i$, we obtain, using Theorem 3.6 of \cite{cs-arscg-89} with $W(x)\equiv x^2$, that

    \begin{eqnarray*}
      E \pth{ \sum_{i=1}^{\tau} \sqrt{m_i} n_i }
      & \leq
      & \sqrt{m}
        \sqrt{ 2 E \pth{ \sum_{i=1}^{\tau} \choosex{n_i}{2} } + E \pth{
        \sum_{i=1}^{\tau} n_i } } \\
      & = & O \pth{ \sqrt{m} \sqrt{ ( r + \frac{r^2}{n^2} w )
            \pth{ \frac{n}{r} }^2} } \\
      & = & O \pth{ \sqrt{m} \sqrt{ \frac{n^2}{r} + w } } \,.
    \end{eqnarray*}

    Thus, substituting in \Eqref{bound}, the expected complexity of the marked faces is

    \begin{equation}
        \eqlab{bound2}%
        O \left( \sum_{i=1}^{\tau} (1+\sqrt{m_i}) n_i \right) \cdot \frac{\lambda_{s+2}(n)}{n} = O \left( n + \frac{r}{n} w + \sqrt{m} \sqrt{ \frac{n^2}{r} + w } \right) \cdot \frac{\lambda_{s+2}(n)}{n} \,.
    \end{equation}

    If we choose $r= \ceil{ \frac{n^{2}}{w} }$, then \Eqref{bound2} becomes
    \[
        O \pth{ \pth{ n+\sqrt{m}\sqrt{w} } \cdot \frac{\lambda_{s+2}(n)}{n} }\,,
    \]
    as asserted.
\end{proof}

\begin{table}[ht]
    \centering%
    \begin{tabular}{|l||c|c|}
     \hline Case
     & Upper bound
     & Lower Bound
     \\
     \hline \hline
     Lines /
     &
       $O \pth{ m^{2/3} n^{2/3} + n } \MakeBigger$
     &
       $\Theta \pth{ m^{2/3}n^{2/3} + n }$
     \\
     \ \ \  Pseudolines
     &
       \cite{cegsw-ccbac-90}
     & \cite{e-acg-87}
     \\
     \hline
         Segments
     &
       $O \pth{ m^{2/3} n^{2/3} + n \alpha(n) + n \log{m} } \MakeBigger$ &
         $\Omega \pth{ m^{2/3}n^{2/3} + n\alpha(n) }$ \\
     & \cite{AEGS92}
     & \cite{EGS90}
     \\
     \hline
     Unit-circles
     &
       $O \pth{ m^{2/3} n^{2/3} \frac{\lambda_{4}(n)}{n} +n } \MakeBigger$
     &
       $\Omega \pth{ m^{2/3}n^{2/3} + n }$
     \\
     &
       \cite{cegsw-ccbac-90}
     &
       \cite{e-acg-87}
     \\
        \hline
     Circles /
     &
       $O \pth{ m^{3/5} n^{4/5} \frac{\lambda_{4}(n)}{n} +n } \MakeBigger$
     &
       $\Omega \pth{ m^{2/3}n^{2/3} + n }$
     \\
     \ \ \ Pseudocircles
     & \cite{cegsw-ccbac-90}
     & \cite{e-acg-87}
     \\
        \hline
         General Curves
       &
         $O \pth{ \sqrt{m} \lambda_{s+2}(n) } \MakeBigger$
       &
         $\Omega \pth{ m^{2/3}n^{2/3} + \lambda_{s+2}(n) }$
     \\
     & \cite{EGPPSS92}, this paper & \cite{EGPPSS92} \\
     \hline
    \end{tabular}

    \caption{The known bounds on the complexity of $m$ faces in arrangements of $n$ curves of some special types.}
   \tbllab{best:bounds}
\end{table}

Note that we always have $w = O(n^2)$, so the above bound is always at most \\ $O \pth{ \sqrt{m} \lambda_{s+2}(n) }$. Thus \thmref{sparse:m:faces} is a strengthening of \thmref{m:Faces:In:Arrangement}, for sparse arrangements.

For the case of segments, a better bound is known:

\begin{theorem}[\cite{AEGS92}]
    The combinatorial complexity of $m$ faces in the arrangement of a set of $n$ line segments with a total of $w$ intersecting pairs is at most
    \[
        O \pth{ m^{2/3} w^{1/3} + n\alpha \pth{ \frac{w}{n} } + n \min \brc{ \log{m}, \log{\frac{w}{n} } }}.
    \]
    \thmlab{sparse:segments}
\end{theorem}

In fact, better bounds hold for sparse arrangements of pseudolines\footnote{A collection of pseudo-lines is a collection of unbounded $x$-monotone Jordan arcs, each pair of which intersect in a single point.}, circles, unit-circles and pseudocircles\footnote{A collection of pseudoCircles is a collection of closed Jordan curves, each pair of which intersect in at most two points, and also each curve intersects any vertical line in at most two points.}. The proof is obtained by plugging in the bounds stated in \tblref{best:bounds} (see \cite{cegsw-ccbac-90}) into the bound (1) in the proof of \thmref{sparse:segments}, and by manipulating the resulting inequalities using Holder`s inequality. We omit the easy technical details and just state the results:

\begin{theorem}
    Let $\Gamma$ be a set of $n$ arcs. Let $w$ be the total complexity of the arrangement $A(\Gamma)$. The complexity of $m$ faces in $A(\Gamma)$ is
    \begin{itemize}
        \item $O \pth{ n \alpha(n) + m^{2/3}w^{1/3} }$ if $\Gamma$ is a collection of pseudolines.

        \item $O \pth{ \lambda_{4}(n ) + m^{2/3}w^{1/3} \frac{\lambda_{4}(n)}{n} }$ if $\Gamma$ is a collection of unit-circles.

        \item $O \pth{ \lambda_{4}(n) + m^{3/5}w^{2/5} \frac{\lambda_{4}(n)}{n} }$ if $\Gamma$ is either a collection of circles, or a collection of pseudocircles.
    \end{itemize}
    \thmlab{sparse:everything}
\end{theorem}

For convenience, the bounds for sparse arrangements are summarized in \tblref{sparse:best:bounds}.

\begin{table}[ht]
    \centering%
   \begin{tabular}{|l||c|c|}
     \hline Case & Upper bound & Lower Bound
     \\
      \hline \hline

      Pseudolines &
      $O \pth{ m^{2/3} w^{1/3} + n \alpha(n) } \MakeBigger$ &
      $\Omega \pth{ m^{2/3}w^{1/3} + n }$ \\
      &
      \thmref{sparse:everything} &
      \remref{sparse:pseudo:lines} \\ \hline

      Segments &
      \begin{tabular}{c}
         $O \left( m^{2/3} w^{1/3} + n \alpha\pth{\frac{w}{n}}
         \right.\MakeBigger\;\;\;\;$ \\
         $\;\;\;\;\;\;\;\;+ \left. n \log{m} + n \log{\frac{w}{n}} \right)$
      \end{tabular} &
      $\Omega \pth{ m^{2/3}w^{1/3} + n\alpha \pth{\frac{w}{n}} }$ \\
      & \cite{AEGS92} & \cite{AEGS92} \\
      \hline

      Unit-circles &
      $O \pth{ m^{2/3} w^{1/3} \frac{\lambda_{4}(n)}{n} +
         \lambda_{4}(n) } \MakeBigger$ &
      $\Omega \pth{ m^{2/3}w^{1/3} + n }$ \\ &
      \thmref{sparse:everything} &
      \remref{sparse:pseudo:lines} \\ \hline

      Circles / &
      $O \pth{ m^{3/5} w^{2/5} \frac{\lambda_{4}(n)}{n} +
         \lambda_{4}(n) } \MakeBigger$ &
      $\Omega \pth{ m^{2/3}w^{1/3} + n }$ \\
     \ \ \ \ Pseudocircles
                  &
                    \thmref{sparse:everything}
                                &
                                   \remref{sparse:pseudo:lines}
     \\%
     \hline

      General Curves &
      $O \left( (n + \sqrt{m}\sqrt{w}) \frac{\lambda_{s+2}(n)}{n}
      \right) \MakeBigger$ &
      $\Omega \pth{ m^{2/3}w^{1/3} + \lambda_{s+2}(n) }$ \\ &
      \thmref{sparse:m:faces} &
      \remref{sparse:pseudo:lines} \\ \hline
   \end{tabular}

   \caption{The known bounds on the complexity of $m$ faces, in a sparse arrangement.}
   \tbllab{sparse:best:bounds}
\end{table}

\begin{remark}
    \remlab{sparse:pseudo:lines}%
    The lower bounds stated in \tblref{sparse:best:bounds} follow by simple modifications to the lower bound construction of \cite{AEGS92}.
\end{remark}

\paragraph*{Acknowledgments.}

I wish to thank Micha Sharir, my thesis advisor, for his help in preparing the paper. He has also contributed to the proof of \thmref{sparse:m:faces}.  My thanks also go to David Zheng for spotting typos in the arXiv version---and for proving me wrong in my assumption that no one would ever look at it.

\BibLatexMode{\printbibliography}

\end{document}